\documentclass{article}

\usepackage[margin=3.2cm]{geometry}
\usepackage{graphicx}
\usepackage[dvipsnames]{xcolor}
\usepackage{soul}
\usepackage{url}
\usepackage{algorithm2e}
\RestyleAlgo{ruled}
\usepackage{comment}
\usepackage{subcaption}
\makeatletter
\let\NAT@parse\undefined
\makeatother
\usepackage{hyperref}
\usepackage{tikz}
\usepackage{bbm}
\usepackage{epsfig}
 
\usepackage{amsmath,amssymb,amsthm}
\newtheorem{theorem}{Theorem}

\theoremstyle{definition}
\newtheorem{assumption}{Assumption}

\definecolor{matt}{rgb}{0.5 0 1}
\setstcolor{matt}

\newcommand{\cA}{\mathcal{A}}
\newcommand{\cB}{\mathcal{B}}
\newcommand{\cC}{\mathcal{C}}
\newcommand{\cD}{\mathcal{D}}

\newcommand{\cP}{\mathcal{P}}

\newcommand{\cX}{\mathcal{X}}

\newcommand{\R}{\mathbb{R}}

\begin{document}

\title{Sequentially learning regions of attraction from data}

\author{Oumayma Khattabi, Matteo Tacchi-Bénard, Sorin Olaru}
\date{May 5, 2025}

\maketitle

\begin{abstract}
The paper is dedicated to data-driven analysis of dynamical systems. It deals with certifying the basin of attraction of a stable equilibrium for an unknown dynamical system. It is supposed that point-wise evaluation of the right-hand side of the ordinary differential equation governing the system is available for a set of points in the state space. Technically, a Piecewise Affine Lyapunov function will be constructed iteratively using an optimisation-based technique for the effective validation of the certificates. As a main contribution, whenever those certificates are violated locally, a refinement of the domain and the associated tessellation is produced, thus leading to an improvement in the description of the domain of attraction.

\end{abstract}

\section{Introduction}
The ongoing rapid increase in data and computing capacity has led to significant advances in control theory, mainly with approaches exploiting data, gaining traction after important achievements in stabilising dynamical systems. Stability analysis follows two main currents: model-based and data-driven. While model-based methods (e.g. model-based predictive control (MPC), adaptive control, Lyapunov direct methods, etc.) have strong stability guarantees, data-driven control laws bypass the model identification step. They are directly building on to real system behaviour, a feature that promotes them as versatile alternatives
in stabilising complex systems. However, the critical issue of control stability certification arises, hindering their applicability.\par

Traditionally, one of the key frameworks used in stability analysis is Lyapunov stability theory \cite{lyapunov1992general}. It is based on finding a function/functional related to the dynamical system, called the Lyapunov candidate, that decreases along trajectories of the system and converges to zero as the system converges to a given attractor (e.g. an equilibrium point). Constructing an adequate Lyapunov function using analytical methods for systems with simple models is moderately easy, but challenging for systems with complex and non-linear models. Some alternative solutions have been proposed; for example, in \cite{julian_parametrization_1999}, the system function in question is represented by a piecewise affine (PWA) approximation for which a PWA Lyapunov candidate is computed using linear programming. In \cite{kim_estimation_2024}, on the other hand, a neural Lyapunov function is constructed to identify positively invariant sets using a PWA activation function. Neural Lyapunov functions are also used in \cite{wang_actor-critic_2024} to certify regions of attraction of the controlled system. There are other methods than constructing Lyapunov functions, like the interval-based method in~\cite{le_mezo_interval_2017} and the interval Lyapunov equation test of \cite{goldsztejn_estimating_2019}, which are also used for stability analysis of complex systems.

Many works in the literature propose data-based methods to prove the stability of the equilibrium point of the system. A prominent method involves the use of deep neural networks. This can be found in \cite{kolter2019learning} where a neural Lyapunov function is learnt along with the system's dynamics to ensure stability guarantees. Another example is \cite{min2023data}, where a neural network-based control law is used and is constrained with a Lyapunov function learnt online. In \cite{grune2020computing}, a neural Lyapunov candidate is constructed while trying to avoid the curse of dimensionality related to neural networks. Another notable stability analysis approach is the continuous piecewise affine (CPA) method by~\cite{hafstein2016computing}, where data points are taken on vertices of a space partition and are used to construct the Lyapunov function. Beyond the approaches above, \cite{martin_data-driven_2024} focuses on Taylor polynomials and utilises piecewise polynomial approximations based on noisy data to study system dissipativity. Other recent research includes \cite{zheng_data-driven_2024}, where a safe stabilising control law is generated along with its Lyapunov function for discrete systems using a noisy dataset.

Existing data-driven stability analysis approaches provide various frameworks for studying system behaviour. This paper will examine one such approach, presented in \cite{TacchiTAC25}, which seeks to find PWA Lyapunov functions for non-linear systems using data. The present contribution is organised as follows: After introducing the mathematical framework, the optimisation problem at the procedure's core is recalled, and as a contribution, an iterative implementation scheme is proposed. The advantage of such a parsimonious treatment of the Lyapunov candidate construction is that it provides patchy Lyapunov functions \cite{ancona2004stabilization} and thus allows for the certification of the domain of attraction in cases that are difficult to validate over the entirety of the region. We propose a technique that targets the uncertified area iteratively and identifies invariant sets based on the learnt Lyapunov functions. Finally, we present a successful resolution of the previous problematic case as a numerical application.

\vspace{0.4cm}
\noindent\textbf{Notation:}
$\R^+$ is defined as the set $\{x\in\R|x\geq0\}$.
For a topological space $\cX$, int$(\cX)$, $\overline{\cX}$, and $\partial\cX$ denote its interior, closure, and boundary, respectively.
A polyhedron is the intersection of a finite number of half-spaces.
A polyunion is a bounded, finite union of polyhedra (not necessarily convex).
The convex hull of a polyunion is denoted as conv(.).
For a set $S$, its cardinal is denoted as $|S|$. This notation is extended for a polyhedron $P$ so that $|P|$ is the number of vertices in $P$.
For $x \in \mathbb{R}^n$ and $r \in \mathbb{R}$, $\mathcal{B}(x,r)$ is used to denote a ball in $\mathbb{R}^n$ of center $x$ and radius $r$.
The Euclidean norm is denoted as $||.||$.
Let $f$ be a function from $\mathcal{X}\subset\mathbb{R}^n$ to $\mathbb{R}^n$, it is Lipschitz continuous, and $M$ is its Lipschitz constant, if for all $x,y \in \mathcal{X}, \ \ ||f(x)-f(y)|| \ \leq \ M||x-y||$.
The strict sub-level set of a function $V$ for a certain level $\alpha \in \mathbb{R}$ is defined as $\mathbf{L}_\alpha^V := \{ x \in \mathbb{R}^n | V(x) < \alpha \}$.
An $N_c$-piece tessellation $\{Y_c\}_{1 \leq c \leq N_c}$ of a subspace $\mathcal{X} \subset \mathbb{R}^n$ is defined as $\bigcup_{c=1}^{N_c} Y_c = \overline{\mathcal{X}}$, such that $\forall c,c' \in \{1,...,N_c\} \ , \ \ c \neq c' \Rightarrow \mathrm{int}(Y_c) \cap \mathrm{int}(Y_{c'}) = \varnothing$, with $Y_c$ convex polyunions which will be called simplexes. In the 2D case, the one-dimensional common facet between two simplices is called an edge, and a singular point of contact between edges is called a vertex. The set of vertices of the tessellation will be denoted $\{v_{i,c}\}^{1\leq i\leq|Y_c|}_{1\leq c\leq N_c}$.

\section{Problem Statement}
This section establishes the theoretical framework necessary for our stability analysis approach. First, we present the fundamental mathematical assumptions to lay the groundwork. Second, we reintroduce the optimisation problem given in \cite{TacchiTAC25} that constitutes the basis of our methodology. Last, we define the objectives and difficulties to tackle.

\subsection{Framework}
Our study addresses the challenge of analysing the domain of attraction of the system with unknown dynamics :
\begin{equation}\label{eq:system}
    \dot{x} = f(x)
\end{equation}
In order to set up a framework that allows a feasible solution, a series of hypotheses have to be imposed with respect to the family of dynamics and the region over which the analysis will be carried out.
\begin{assumption}\label{as:system}
  The system under study has a minimal state space realisation. For such a minimal state space realisation, the following hold:
  \begin{itemize}
      \item the dimension $n$ of the state vector is known;
      \item the right-hand side $f(.)$ of \eqref{eq:system} is accessible for evaluation either via measurements or through simulation;
      \item the function $f:\R^n\rightarrow \R^n$ is Lipschitz continuous, and an upper bound $M$ for its Lipschitz constant is known a priori.
  \end{itemize}
\end{assumption}

Analysing dynamical systems can reveal complex behaviours, including multiple equilibria and chaotic trajectories. To maintain focus within our study’s scope, we introduce additional constraints defining the solution space and the qualitative/quantitative structure of the solutions to~\eqref{eq:system}.

\begin{assumption}\label{as:set}
The analysis will be carried on on a bounded set $\mathcal{X}\subset\mathbb{R}^n$ of the state space. It is assumed that there is one locally asymptotically stable equilibrium $\bar{x} = 0 \in int(\mathcal{X})$, and no other equilibrium point exists in $\overline{\mathcal{X}}$.
\end{assumption}

Based on Assumption~\ref{as:system}, we can have access to a dataset $\mathcal{D}:=\{(x_d,f_d)\}_{1\leq d\leq N_d}$, with $f_d=f(x_d)$ and $N_d$ the number of data points, and we can choose the position of the state samples $x_d$ when generating data.

The aim is to be able to certify the stability of system \eqref{eq:system} in $\mathcal{X}$ based on the dataset $\mathcal{D}$ and the Lipschitz constant $M$.

\subsection{Data-driven Lyapunov stability analysis}
This work is based on the method presented by Tacchi et al. in \cite{TacchiTAC25}, in which we identify regions of attraction by constructing certified and robust Lyapunov functions consistent with the available data. The following Assumptions provide two additional technical conditions to avoid singularities in the numerical implementation:
\begin{assumption} \label{as:data}
The following information is available with respect to \eqref{as:system}:
\begin{itemize}
    \item a neighbourhood $\cA$ of the equilibrium point is known to be attracted to $\bar{x} = 0$ ;
    \item a fixed dataset $\cD$ of the system is available ;
    \item $\overline{\cX}$ and $\cA$ are polyunions ;
\end{itemize}
\end{assumption}

Let the region $\overline{\cX\setminus\cA}$ be subdivided into an $N_c$-piece tessellation $\{Y_c\}_{1 \leq c \leq N_c}$, with $N_c$ the number of cells $Y_c$. We set each cell $Y_c$ to be a convex polyunion with vertices $\{v_{1,c},\ldots,v_{|Y_c|,c}\}.$ To each cell $Y_c$, an affine Lyapunov function $V_c(x) := g_c^\top x + b_c$ is assigned. These functions are continuously merged together and constitute one global PWA Lyapunov candidate $V$ on $\mathcal{X}$ defined as follows :
\begin{equation}\label{eq:Lyap}
    \forall c \in \{1,...,N_c\}, \ \forall x \in Y_c \subset \overline{\cX\setminus\cA}, \ \ V(x) = g_c^\top x+b_c
\end{equation}
Based on these elements, and according to~\cite[Proposition 6]{TacchiTAC25}, if the upper bound on $\nabla V^\top f$ based on data is negative at all points in $\overline{\cX\setminus\cA}$ where $\nabla V$ exists, then for all $\alpha\in\R$ such that $\mathbf{L}_\alpha^V\subset\cX$, $\mathbf{L}_\alpha^V$ is include in the region of attraction.\par

The following optimisation problem is a formulation of the conditions to find a PWA Lyapunov function $V$ :
\begin{subequations}\label{eq:optprob}
    \begin{align}
        \mathbf{s}_{\alpha,\mu} := \min_{\{\gamma_{d,c},g_c,b_c,s_{i,c}\}} \ \sum_{c=1}^{N_c} \sum_{i=1}^{|Y_c|}& s_{i,c} \label{eq:obj} \\
        \forall \ c,c' \in \{{\scriptstyle 1,...,N_c}\}, \ i \in \{{\scriptstyle 1,...,|Y_c|}\},\ \ s_{i,c} &\geq -\mu \label{eq:cnd1} \\
        v_{i,c} \in Y_{c'} \ \Rightarrow \ (g_c-g_{c'})^\top v_{i,c} &= b_{c'}-b_c \label{eq:cnd2} \\
        \sum_{d=1}^{N_d} \gamma_{d,c} &= g_c \label{eq:cnd5} \\
        \sum_{d=1}^{N_d} \gamma_{d,c}^\top f_d + ||\gamma_{d,c}||M||v_{i,c}-x_d|| &\leq s_{i,c} \label{eq:cnd6}
    \end{align}
\end{subequations}
with $\{s_{i,c}\}^{1\leq i\leq|Y_c|}_{1\leq c\leq N_c}$ some slack variables to minimise (one per vertex of each cell of the tessellation), $\{\gamma_{d,c}\}^{1\leq d\leq N_d}_{1\leq c\leq N_c}$ data dependent variables that tune $\{g_c\}_{1\leq c\leq N_c}$ using constraints \eqref{eq:cnd5}, and $\mu\in\R^+$ defining the lower bound of the slack variables. The constraints \eqref{eq:cnd2} represent the continuity of the Lyapunov function $V$ between neighbouring cells $Y_c$ and $Y_{c'}$. The negativity of $\nabla V^\top f$ is ensured by constraint \eqref{eq:cnd6}. Once the function is obtained, a level $\alpha$ should be found such that $\mathbf{L}^V_\alpha \subset \cX$ with $\cA \subset \mathbf{L}^V_\alpha$, proving that $\mathbf{L}^V_\alpha$ is attracted to the equilibrium.
The learned function $V$ is a certified Lyapunov function if all slack variables $\{s_{i,c}\}^{1\leq i\leq|Y_c|}_{1\leq c\leq N_c}$ are negative. The negativity of $\nabla V^\top f$ is thus guaranteed, making the resulting $V$ a valid Lyapunov function.

To avoid the problem of infeasibility, \cite{TacchiTAC25} provides a necessary condition to be validated by the dataset :
\begin{equation}
    \{v_{i,c}\}^{1\leq i\leq|Y_c|}_{1\leq c\leq N_c} \subset \bigcup_{d=1}^{N_d} \cB\left(x_d,r_d=\frac{||f_d||}{M}\right)
    \label{eq:datacond}
\end{equation}
This implies that data should provide at least some minimal information relevant to each vertex of the tessellation.

\subsection{Advantages, challenges and objectives}
This method is a promising alternative to the state of the art, where the common practice is to use semi-definite programming (SDP)~\cite{martin_data-driven_2024,zheng_data-driven_2024}. In contrast, the optimisation problem~\eqref{eq:optprob} is based on second-order cone programming (SOCP) and linear constraints, which makes it more scalable to higher dimensions. Also, unlike neural networks that train on big datasets, it does not need ample datasets to give satisfying results. Moreover, there is no binding relation between the position of the data points and the tessellation except for the learnability condition~\eqref{eq:datacond}, which is quite flexible. This gives more freedom to adjust the dataset and vertices depending on the system at hand.

It should be mentioned that optimisation problem~\eqref{eq:optprob} gives satisfying results when applied to various nonlinear systems with nonconvex regions of attraction, like the Van Der Pol oscillator (see numerical results in~\cite{TacchiTAC25}). In some cases, constraints \eqref{eq:cnd1} are not saturated, but all slack variables are negative, which is enough to certify attractivity. However, there are exceptional cases where numerical computations provide (according to the chosen tessellation) a Lyapunov function certified in only part of the set considered initially. An example of one such case is the  damped pendulum with the following state space representation:
\begin{equation}
    \dot{x} = \left(\begin{array}{c} \dot{\theta} \\ \dot{\omega} \end{array}\right) = \left(\begin{array}{c} \omega \\ -\sin(\theta)-2\omega \end{array}\right)
    \label{eq:pendsys}
\end{equation}

For a randomly chosen Delaunay triangulation and dataset as displayed in Figure~\ref{sfig:pendata}, the stability is, on average, not certified by the algorithms of~\cite{TacchiTAC25} in the diagonal central part of the region (see figure \ref{sfig:pendlyap}), even though it is attracted to the equilibrium. The expected result was a Lyapunov function covering the entirety of the tessellation. However, the optimisation problem could not provide negative slack variables and thus valid affine functions in some of the simplexes. These simplexes get eliminated as they don't validate the Lyapunov function criteria, which leaves a considerable area of the initial set uncovered and thus uncertified by the PWA Lyapunov candidate. This result is hypothesised to be caused by the drastic variation of the values of $\{f_d\}_{1\leq d\leq Nd}$ (represented by orange arrows in~\ref{sfig:pendata}), where data points with larger image values (big arrows in the corners) exert significant influence on the optimisation process but fail to sufficiently contribute information on the images of the vertices in areas where $f_d$ takes on smaller absolute values. Consequently, these areas with lower function absolute values are under-represented and inaccurately characterised, eventually affecting the overall accuracy of the optimisation.

\begin{figure}[htbp]
    \centering
    \begin{subfigure}{0.45\linewidth}
        \includegraphics[width=\linewidth]{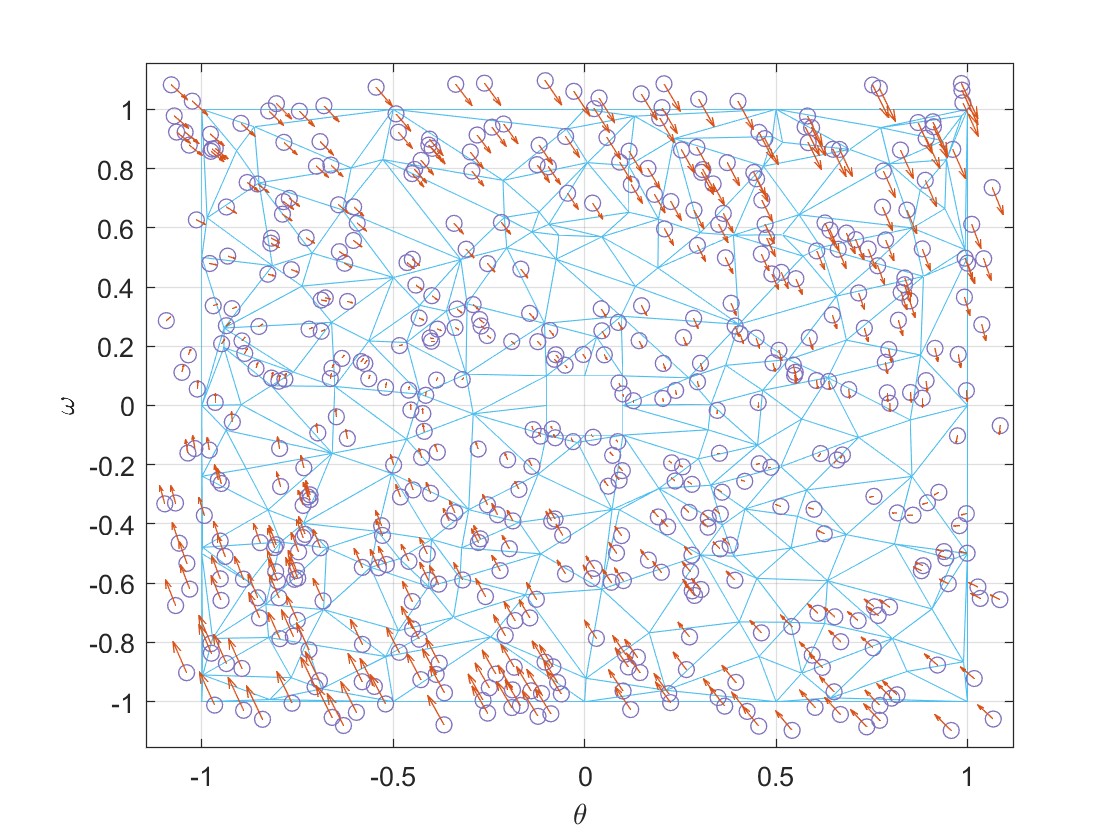}
        \caption{\scriptsize{Data and tessellation}}
        \label{sfig:pendata}
    \end{subfigure}
    \hspace{0.05\linewidth}
    \begin{subfigure}{0.45\linewidth}
        \includegraphics[width=\linewidth]{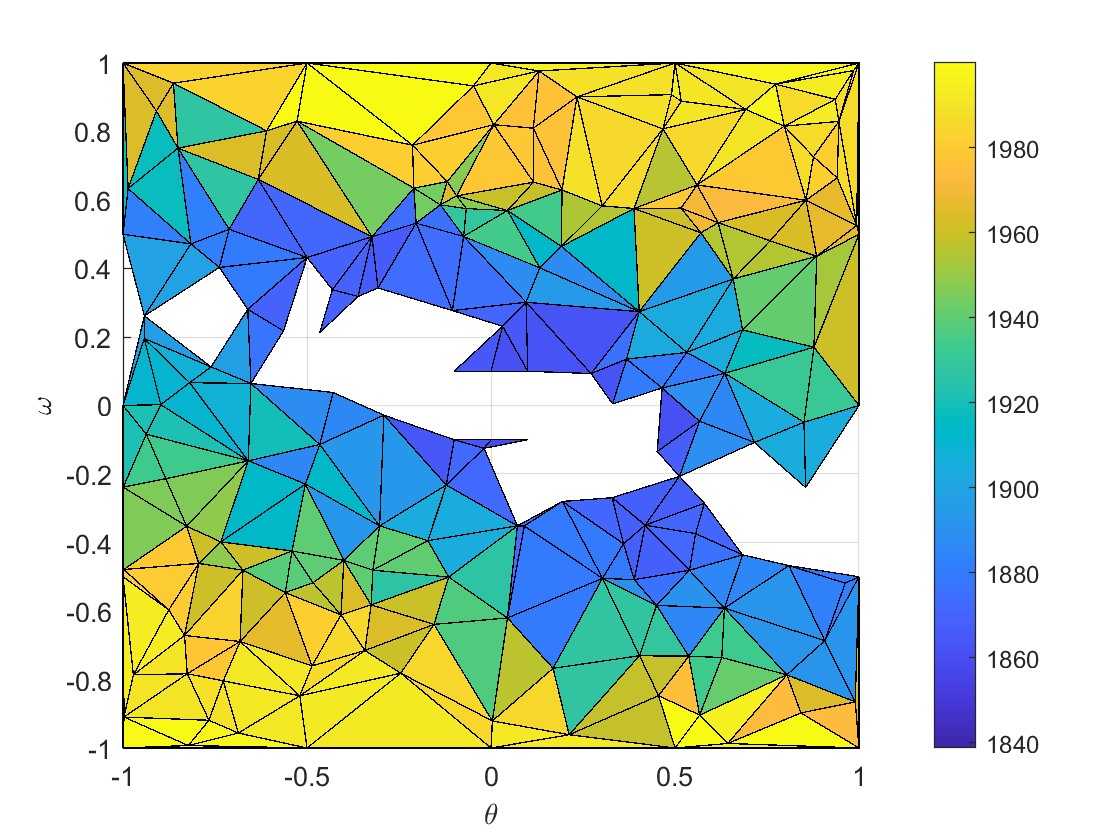}
        \caption{\scriptsize{Result of the optimisation}}
        \label{sfig:pendlyap}
    \end{subfigure}
    \caption{Example case of the simple pendulum}
    \label{fig:pend}
\end{figure}

According to~\cite[Appendix B]{TacchiTAC25}, a dataset and a tessellation providing enough covering of the set $\overline{\cX\setminus\cA}$ are bound to bring forth a certified Lyapunov function. However, there is no sufficient condition for the density and distribution of the data and vertices to allow the algorithm's convergence to optimality. Moreover, the complexity of the computations explodes as the number of constraints is of the order of $\mathcal{O}((3|Y_c|+1)\cdot N_c)$ and the number of optimisation variables is around $\mathcal{O}((|Y_c|+2N_d+3)\cdot N_c)$, which makes it difficult to boundlessly increase $N_d$ and $N_c$ until optimality is reached.

The objective of this work is to solve the non-certifiability of the Lyapunov function in parts of the region and to overcome the complexity issue to ensure that solving this optimisation problem can effectively establish the stability of the equilibrium point of the system in $\mathcal{X}$, using the example of the pendulum of system \eqref{eq:pendsys} as validation.

\section{Proposed Approach}
To overcome the challenges encountered when implementing this stability analysis method and demonstrate the practical applicability of the mathematical theory, we propose in this section an iterative approach using the optimisation problem, where we iteratively adapt the data and tessellation to deploy to get a Lyapunov function certifying the largest region of attraction to the equilibrium point in the set $\cX$.

\subsection{The principle of the iterative construction}
Since one of the core issues lies in the uncertainty regarding the existence of a PWA function solution for the optimisation problem with all slacks negative, a first sanity check is needed to ascertain that a PWA Lyapunov candidate exists. Note that this sanity check uses the full knowledge of the system model~\eqref{eq:pendsys}, while usually, the dynamics $f$ are unknown. However, our iterative construction still applies in the unknown setting, provided one can refine their tessellation and data distribution.

For this, we consider a quadratic Lyapunov function $V_Q$ of the system \eqref{eq:pendsys}:
\begin{equation}
    \forall x \in \R^2 \ \text{ s.t. } \ ||x||_\infty\leq1 \ , \quad V_Q(x) = \frac{1}{2}x^TQx
\end{equation}
with $Q = \left(\begin{array}{cc}
    10 & 3 \\
    3 & 4
\end{array}\right)$.

Next, a PWA approximation $\hat{V}_Q$ of $V_Q$ is computed.  If $\hat{V}_Q$ is a Lyapunov candidate, then we have proven the existence of a PWA Lyapunov function and a tessellation/subdivision $\{Y_c\}_{1 \leq c \leq N_c}$ corresponding to the affine pieces of $\hat{V}_Q$.

Since we know the model and a PWA Lyapunov function here, we can use the corresponding subdivision as our initial tessellation for faster convergence. In the general case, when the dynamics are unknown, using a Delaunay triangulation also works \cite{delaunay1934sphere}, but it only takes more iterations to converge. Next, we generate data points that offer a covering of the tessellation so that learnability condition~\eqref{eq:datacond} is met. Then, we solve optimisation problem~\eqref{eq:optprob} for the constructed inputs, obtain a Lyapunov candidate $V$, and look for non-negative slack variables. If none exist, then the largest $\alpha$ with $\cA \subset \mathbf{L}^V_\alpha \subset \cX$ defines a region attracted to $0$ along system trajectories. Else, we identify the uncertified points $$\mathcal{N} = \{ (i,c) \in \mathbb{N}^2 \mid 1 \leq c \leq N_c, 1 \leq i \leq |Y_c|, s_{i,c} \geq 0 \}$$
from which we deduce a convex uncertified region: 
\begin{equation} \label{eq:uncertif}
\cC = \mathrm{conv}(\cA \cup \{v_{i,c}\}_{(i,c) \in \mathcal{N}}).
\end{equation}
By design of constraint~\eqref{eq:cnd6} (see~\cite[Appendix A]{TacchiTAC25} for details), we can certify that $\forall x \in \cX \setminus \cC$, $\nabla V(x)^\top f(x) < 0$. The problem is that, contrary to the setting in~\cite[Appendix A]{TacchiTAC25}, the case $\cC \nsubseteq \cA$ can occur, preventing us from concluding anything of significance on whether $\cC$ is part of the region of attraction to the equilibrium. In that case, one can study the stability in $\cC$ by iterating as follows:
\begin{itemize}
    \item 
    suppose a first iteration $k$ failed to compute $V_k$, $\alpha_k$ and $\cC_k$ based on $\cX_k$, $\cA_k$ and $\cD_k$ such that $\cC_k \nsubseteq \cA$. Define $\cX_{k+1} = \beta_{k+1}\cdot \cC_k$, where $\beta_{k+1} \geq 1$ is the highest upscaling parameter with $\cX_{k+1} \subset \mathbf{L}^{V_k}_{\alpha_k}$ still verified; if no such $\beta_{k+1} \geq 1$ exists, it means that the uncertified area intersects or is outside of the level set limit where $V_k = \alpha_k$, meaning this particular iteration failed and should be restarted with a finer tessellation and more data points until a better and valid result is generated.
    \item Introduce a smaller $\cA_{k+1} \subset \cA_k \subset \cC_k$ ($\cA_0 = \cA$ ensures all such sets are attracted to the equilibrium, and non-negative slacks are often in the neighbourhood of the exclusion zone $\cA_k$, so taking a smaller $\cA_{k+1}$ can validate the condition $\cC_{k+1}$ in $\cA$).
    \item Compute a tessellation of $\cP = \cX_{k+1} \setminus \cA_{k+1}$ by adding vertices to make it denser, and refine the dataset $\cD_k$ into some $\cD_{k+1}$ accordingly by removing some data points and adding others while maintaining learnability condition~\eqref{eq:datacond}.
    \item Solve optimisation problem~\eqref{eq:optprob} for the new input to obtain a new $V_{k+1}$. Find $\alpha_{k+1}$ such that $\cC_k \subset \mathbf{L}^{V_{k+1}}_{\alpha_{k+1}} \subset \cX_{k+1}$; if no $\alpha_{k+1}$ such that $\cC_k \subset \mathbf{L}^{V_{k+1}}_{\alpha_{k+1}}$ exists, the new Lyapunov candidate fails to cover all uncertified points and the iteration fails. 
    This iteration needs to be done again in this case with refined tessellation and richer data. 
    Otherwise, define $\cC_{k+1}$ and proceed.
    \item If $\cC_k \subset \cA$ or $k = k_{\max}$ (max number of iterations reached), stop iterating; else, start again from the first step with $k \longleftarrow k+1$.
\end{itemize}

The iterative method consists of repeating all the steps, from choosing the parameters, data and tessellation to testing if the uncertified area is contained in $\cA$, until this last condition is verified. It works following Algorithm~\ref{alg}.
\begin{algorithm}
    \caption{}\label{alg}
    \textbf{Input:} \textit{$\mathcal{X}$, $\mathcal{A}$, $\mu$, $M$}\\
    \textbf{Output:} \textit{$\{\mathbf{L}^{V_k}_{\alpha_k}\}$}\\
    $\cP\gets\cX\setminus\cA$\\
    \While{number of iterations $k$ $<$ $k_{\max}$}
    {Generate tessellation $\{Y_c\}_{1 \leq c \leq N_{c_k}}$ with $\cP=\cup_{c=1}^{N_{c_k}} Y_c$\\
    Generate dataset $\cD_k$ verifying condition \eqref{eq:datacond}\\
    Run optimisation problem \eqref{eq:optprob} and obtain $V_k$\\
    \eIf{all slack variables are negative}
    {Break}
    {Identify the non-certified polyunion $\cC_k$\\
    Identify the corresponding level set $\mathbf{L}^{V_k}_{\alpha_k}$ of $V_k$\\
    \uIf{$\cC_{k-1}\nsubseteq\mathbf{L}^{V_k}_{\alpha_k}$}{Iteration failed, restart with a more refined tessellation and dataset}
    \uElseIf{$\cC_k\subseteq\cA_k$}
    {Break}
    \uElseIf{$\beta_{k+1}$ exists such that $\beta_{k+1}\cdot \cC_k \subset \mathbf{L}^{V_k}_{\alpha_k}$}{Choose $\cA_{k+1}\subset\cA_k$\\
    Accept $\beta_{k+1}$ and $\cX_{k+1} = \beta_{k+1}\cdot \cC_k$\\
    $\cP\gets \cX_{k+1}\setminus\cA_{k+1}$}
    }}
\end{algorithm}

Algorithm~\ref{alg} comes with a certification result, which we state in Theorem~\ref{thm}.

\begin{theorem}\label{thm}
    Let $\cX = \cX_0 \supset \cX_1 \supset \ldots \supset \cX_K$, $\cA = \cA_0 \supset \cA_1 \supset \ldots \supset \cA_K$, $\cD = \cD_0, \ldots, \cD_K$ defining ambient sets, prior attracted sets and datasets of instances of problem~\eqref{eq:optprob}. Let $V_0, \ldots, V_K$ be the patchy PWA Lyapunov candidates associated to feasible solutions of the corresponding problem instances, with uncertified regions $\cC_0, \ldots, \cC_K$ and levels $\alpha_0, \ldots, \alpha_K$ (in particular, $\forall k \in \{0,\ldots,K\}$, $\mathbf{L}^{V_k}_{\alpha_k} \subset \cX_k$). Suppose that:
    \begin{itemize}
        \item $\forall k \in \{1,\ldots,K\}$, $\cC_{k-1} \subset \mathbf{L}^{V_k}_{\alpha_k} \subset \cX_k \subset \mathbf{L}^{V_{k-1}}_{\alpha_{k-1}}$ 
        \item $\cC_K = \cA_K$ (no non-negative slack variable) or $\cC_K \subset \cA_0$
    \end{itemize}
    Then, $\mathbf{L}^{V_0}_{\alpha_0}$ is attracted to $0$ along the unknown dynamics, i.e.
    \begin{equation} \label{eq:thm}
        V_0(x(0)) \leq \alpha_0 \quad \Longrightarrow \quad x(t) \underset{t\to\infty}{\longrightarrow} 0.
    \end{equation}
\end{theorem}

\begin{proof}
    Suppose that $K = 0$. Then, $\cC_0 = \cA_0$ i.e. all slack variables in problem~\eqref{eq:optprob} are negative, and ~\cite[Theorem 8]{TacchiTAC25} ensures that~\eqref{eq:thm} holds.
    
    Suppose that $V_k(x(0)) \leq \alpha_k \Longrightarrow \lim_{t\to\infty}x(t) = 0$. Then, since $\cC_{k-1} \subset \mathbf{L}^{V_k}_{\alpha_k} \subset \cX_{k-1}$ and $V_{k-1}$ is a feasible solution for problem~\eqref{eq:optprob},
    $\nabla V_{k-1}(x)^\top f(x) < 0$ is certified for all $x \in \cX_{k-1} \setminus \mathbf{L}^{V_k}_{\alpha_k}$ and according to~\cite[Proposition 6]{TacchiTAC25}, it holds $V_{k-1}(x(0)) \leq \alpha_{k-1} \Longrightarrow \lim_{t\to\infty} x(t) = 0$.

    By induction, since $\mathbf{L}^{V_0}_{\alpha_0} \supset \mathbf{L}^{V_1}_{\alpha_1} \supset \ldots \supset \mathbf{L}^{V_K}_{\alpha_K}$, it is therefore enough to prove that 
    $$V_K(x(0)) \leq \alpha_K \Longrightarrow \lim_{t\to\infty}x(t) = 0,$$
    which is done by noticing that, since $\cC_K \subset \cA_0$ (because $\cC_K \subseteq \cA_k \subset \cA_0$), $\nabla V_K(x)^\top f(x) < 0$ is certified for all $x \in \cX_K \cap (\cX_0 \setminus \cA_0)$ and by reusing~\cite[Proposition 6]{TacchiTAC25}.
\end{proof}

\subsection{Implementation and practical challenges}
In \cite{TacchiTAC25}, the last algorithm proposed (Algorithm 3) bears some similarity to our proposed Algorithm~\ref{alg} as it is also an iterative approach, but they are fundamentally different. In \cite[Algorithm 3]{TacchiTAC25}, the iterations start from the smallest region to the largest to build one continuous Lyapunov function covering the union of the areas. It also stops if it cannot generate a certified Lyapunov function in the region. Our algorithm, on the other hand, starts from the global region and certifies a region of attraction towards an inner region. If the inner region is not included in $\cA$, it repeats the same step while covering smaller and smaller areas in each iteration. The goal is not to build a global Lyapunov function, but to prove the region's attraction iteratively by constructing patchy Lyapunov functions in the region such that each function defines an invariant set based on its sublevel set.

To provide greater flexibility in the optimisation problem and ensure the validation of continuity constraints, all the tessellations are taken so that the cells are triangular. As for the dataset, considering that the region that proved difficult for the algorithm has data that can't be informative enough for all the vertices, we include data from neighbouring areas outside of the tessellation until condition \eqref{eq:datacond} is satisfied. Some data points are also taken on the vertices for even better coverage.

\section{Numerical Results}
In this part, we present the implementation of the iterative method of the optimisation problem described in Algorithm~\ref{alg} on the damped pendulum with the mathematical model~\eqref{eq:pendsys}. The tessellation used is ensured to contain a PWA Lyapunov function. We seek to retrieve the attractivity of the origin 
as the equilibrium in the chosen region using Lyapunov candidates computed from data measurements, ignoring the full model~\eqref{eq:pendsys}. It should be mentioned that even if the next simulation results shown are informed of the model solely in establishing the tessellation, a simulation that respects the assumption of the unavailability of the model at all times is also successful, albeit it needs more iterations. The results are shown for comparison in Figure~\ref {fig:levelsets}.

For the first implementation, we use the parameters described in Table \ref{tab:init_params} such that the initial set $\mathcal{X}$ is bounded by $l_\mathcal{X}$ and the neighbourhood $\mathcal{A}$ of the equilibrium point is bounded by $l_\mathcal{A}$.

\begin{table}[htbp]
    \caption{Initial optimization parameters}
    \centering
    \begin{tabular}{|c|c|}
        \hline
        \textbf{Parameter} & \textbf{Value} \\
        \hline
        $l_\mathcal{X}$ & $[-1,1;-1,1]$ \\
        $l_\mathcal{A}$ & $0.1.l_\mathcal{X}$ \\
        $\mu$ & 5\\
        $M$ & $2.5$ \\
        \hline
    \end{tabular}
    \label{tab:init_params}
\end{table}

The algorithm was \href{https://github.com/OumaymaK/Sequential_data-driven_stability_analysis}{coded} using YALMIP and MPT3 toolboxes \cite{Lofberg2004, MPT3}, was run on a laptop with 32GB memory and AMD Ryzen 7 7735U, and the solver used was MOSEK. The problem was solved in five iterations. Figures \ref{fig:iter1} to \ref{fig:iter5} show the results of each iteration.

\begin{figure}[htbp!]
    \centering
    \begin{subfigure}{.42\linewidth}
        \includegraphics[width=\linewidth]{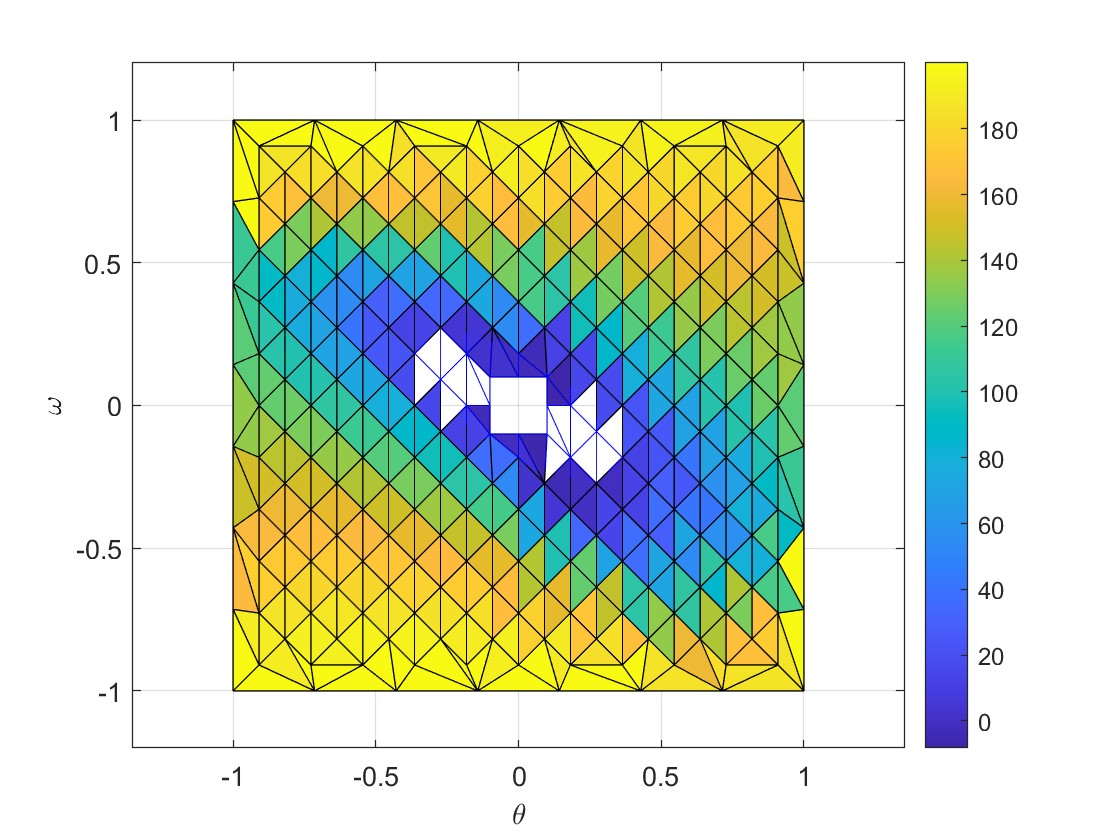}
        \caption{\tiny $N_{d_k}=300$, $N_{v_k}=252$}
    \end{subfigure}
    \begin{subfigure}{.42\linewidth}
        \includegraphics[width=\linewidth]{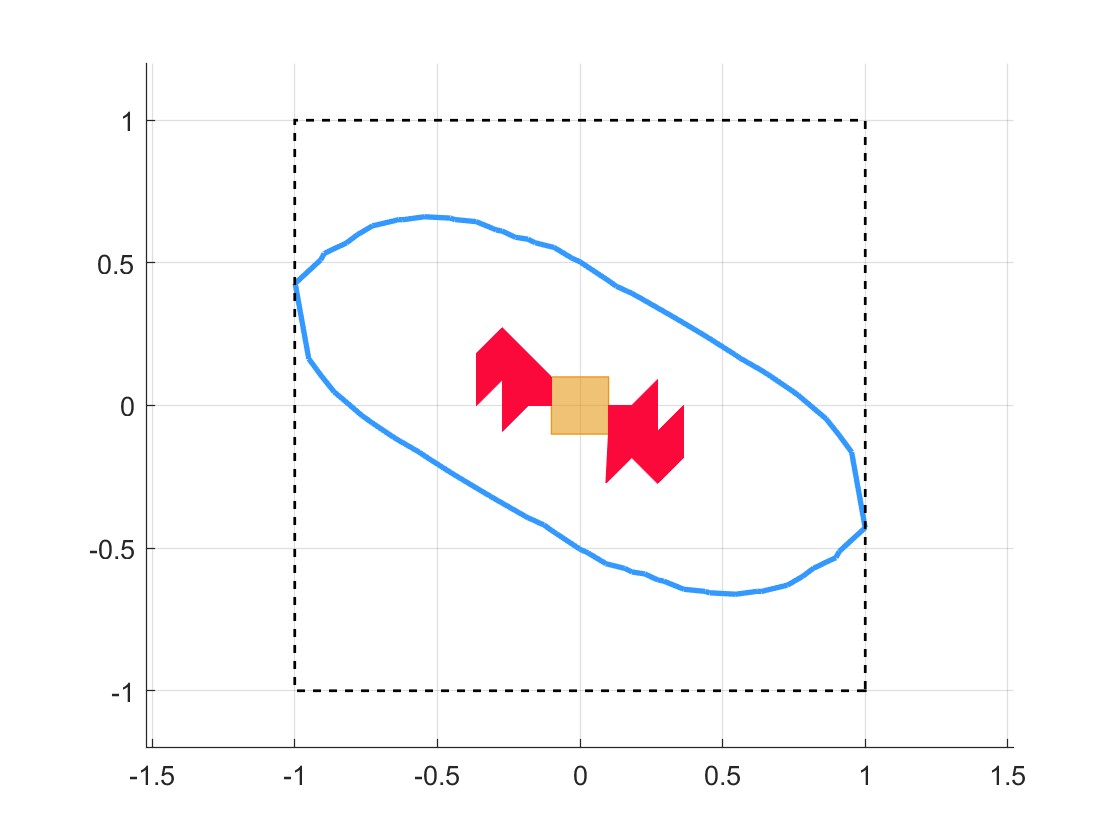}
        \caption{\tiny $\alpha_k = 111$}
    \end{subfigure}
    \caption{$k=0$}
    \label{fig:iter1}
\end{figure}

\begin{figure}[htbp!]
    \centering
    \begin{subfigure}{.42\linewidth}
        \includegraphics[width=\linewidth]{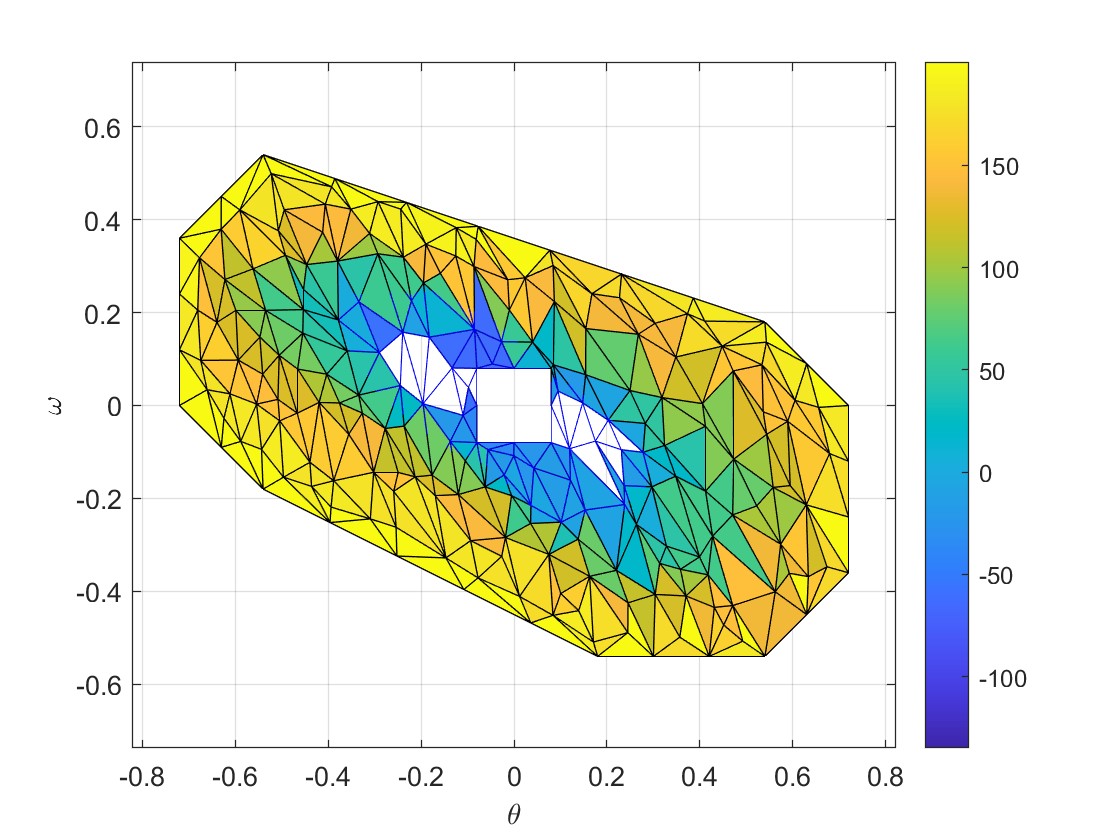}
        \caption{\tiny $N_{d_k}=349$, $N_{v_k}=221$}
    \end{subfigure}
    \begin{subfigure}{.42\linewidth}
        \includegraphics[width=\linewidth]{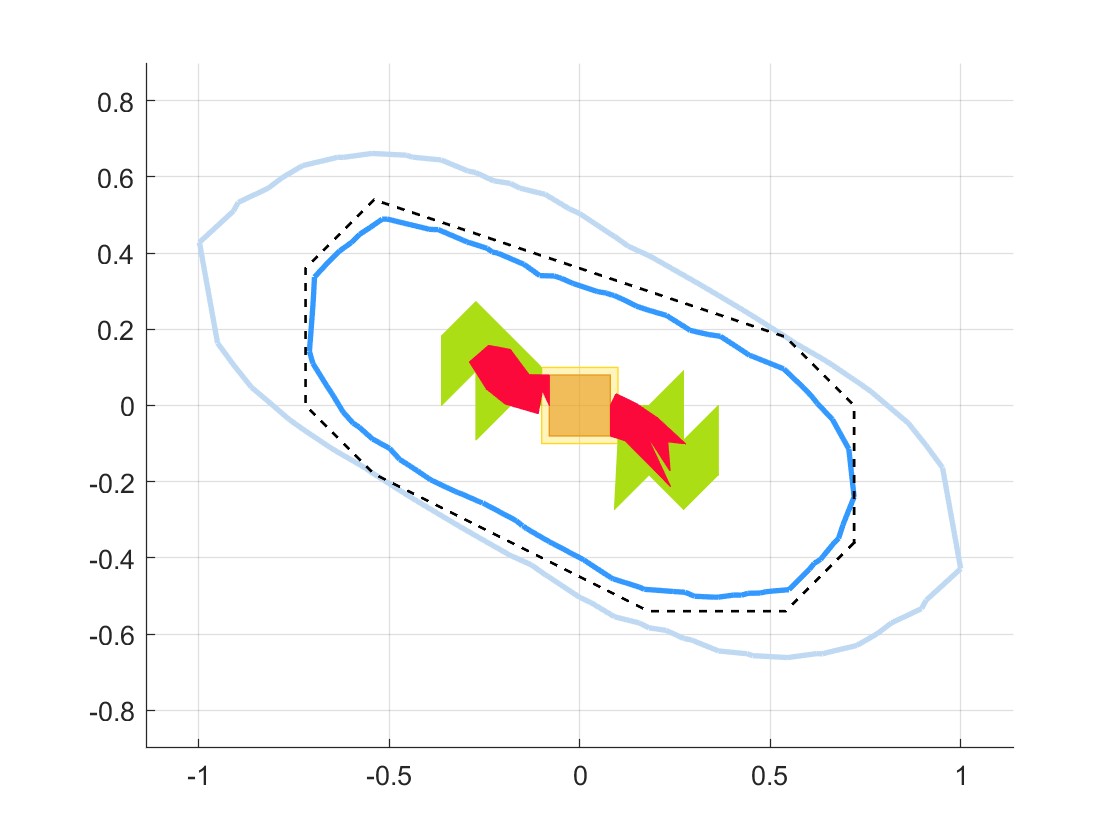}
        \caption{\tiny $\alpha_k = 174$, $\beta_k=1.98$}
    \end{subfigure}
    \caption{$k=1$}
    \label{fig:iter2}
\end{figure}

\begin{figure}[htbp!]
    \centering
    \begin{subfigure}{.42\linewidth}
        \includegraphics[width=\linewidth]{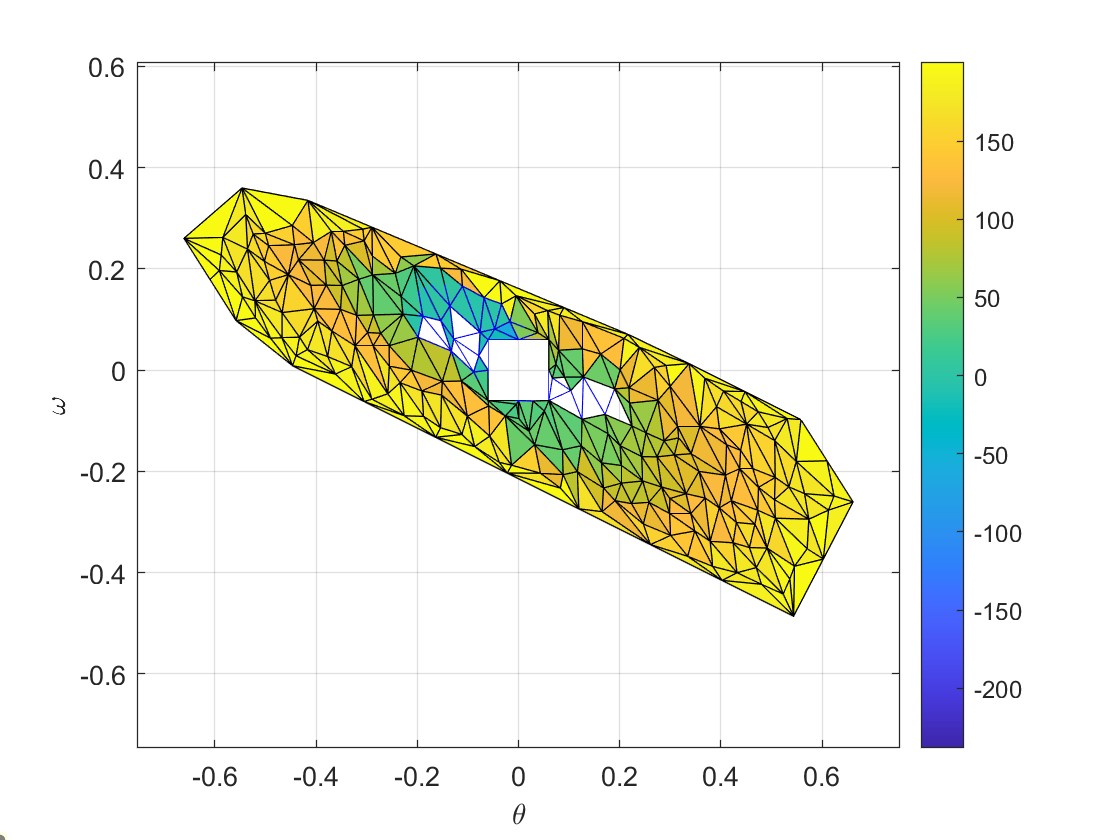}
        \caption{\tiny $N_{d_k}=346$, $N_{v_k}=238$}
    \end{subfigure}
    \begin{subfigure}{.42\linewidth}
        \includegraphics[width=\linewidth]{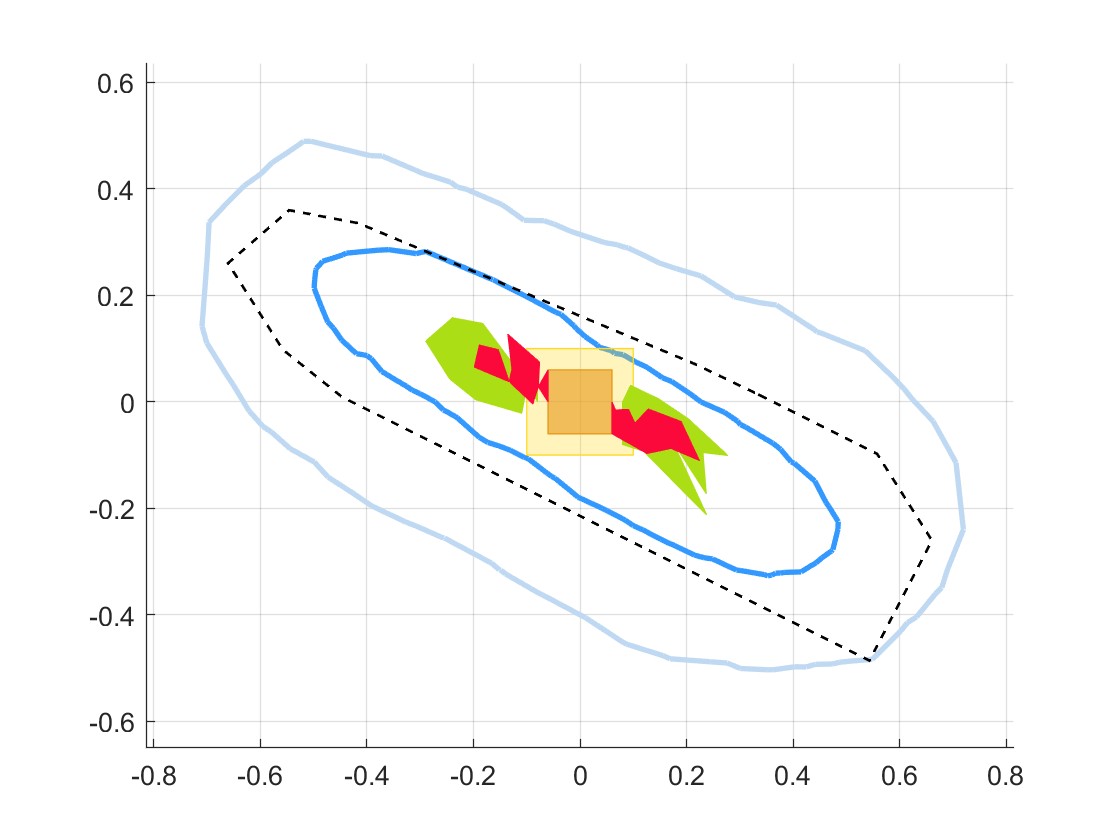}
        \caption{\tiny $\alpha_k = 146$, $\beta_k=2.26$}
    \end{subfigure}
    \caption{$k=2$}
    \label{fig:iter3}
\end{figure}

\begin{figure}[htbp!]
    \centering
    \begin{subfigure}{.42\linewidth}
        \includegraphics[width=\linewidth]{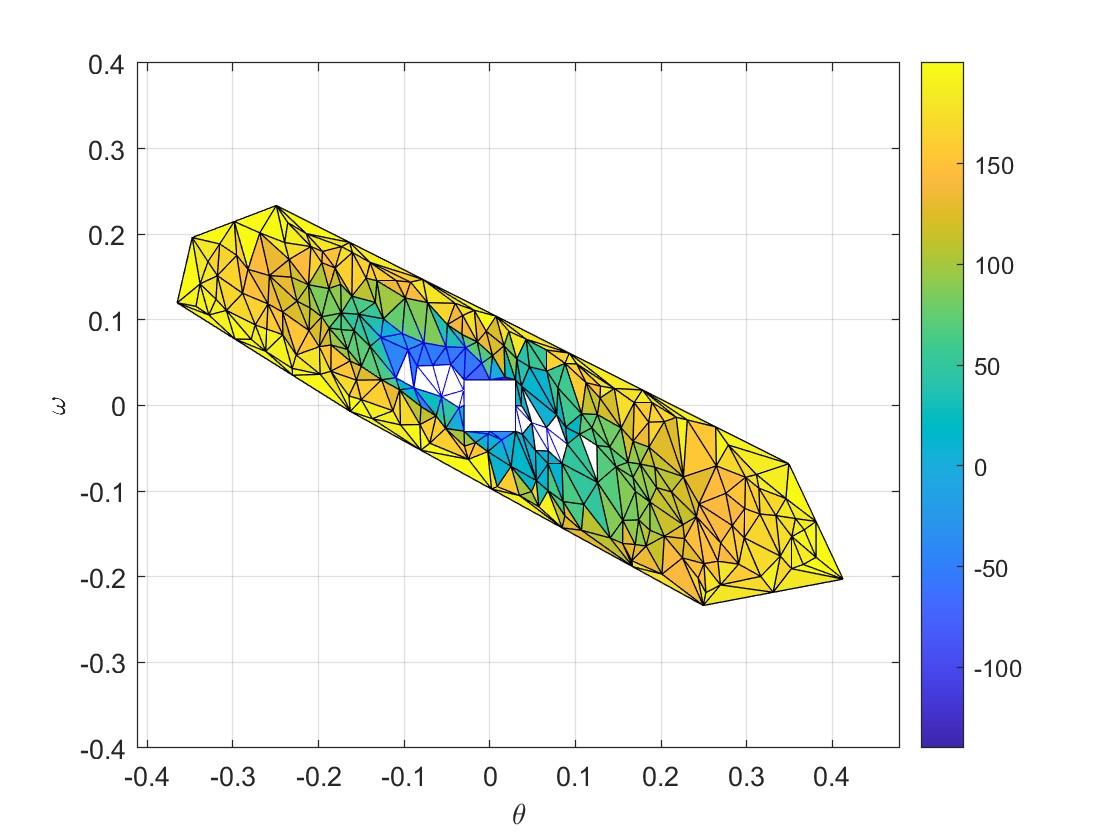}
        \caption{\tiny $N_{d_k}=348$, $N_{v_k}=241$}
    \end{subfigure}
    \begin{subfigure}{.42\linewidth}
        \includegraphics[width=\linewidth]{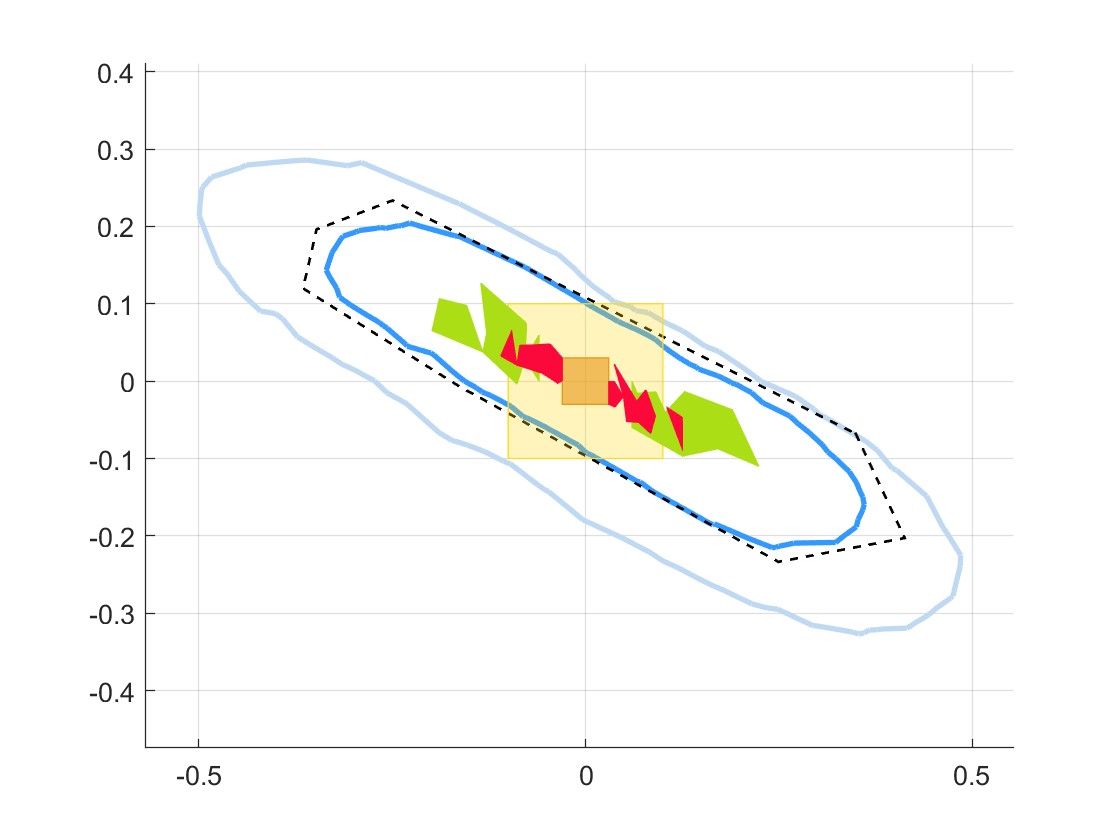}
        \caption{\tiny $\alpha_k = 180$, $\beta_k=1.8$}
    \end{subfigure}
    \caption{$k=3$}
    \label{fig:iter4}
\end{figure}

\begin{figure}[htbp!]
    \centering
    \begin{subfigure}{.42\linewidth}
        \includegraphics[width=\linewidth]{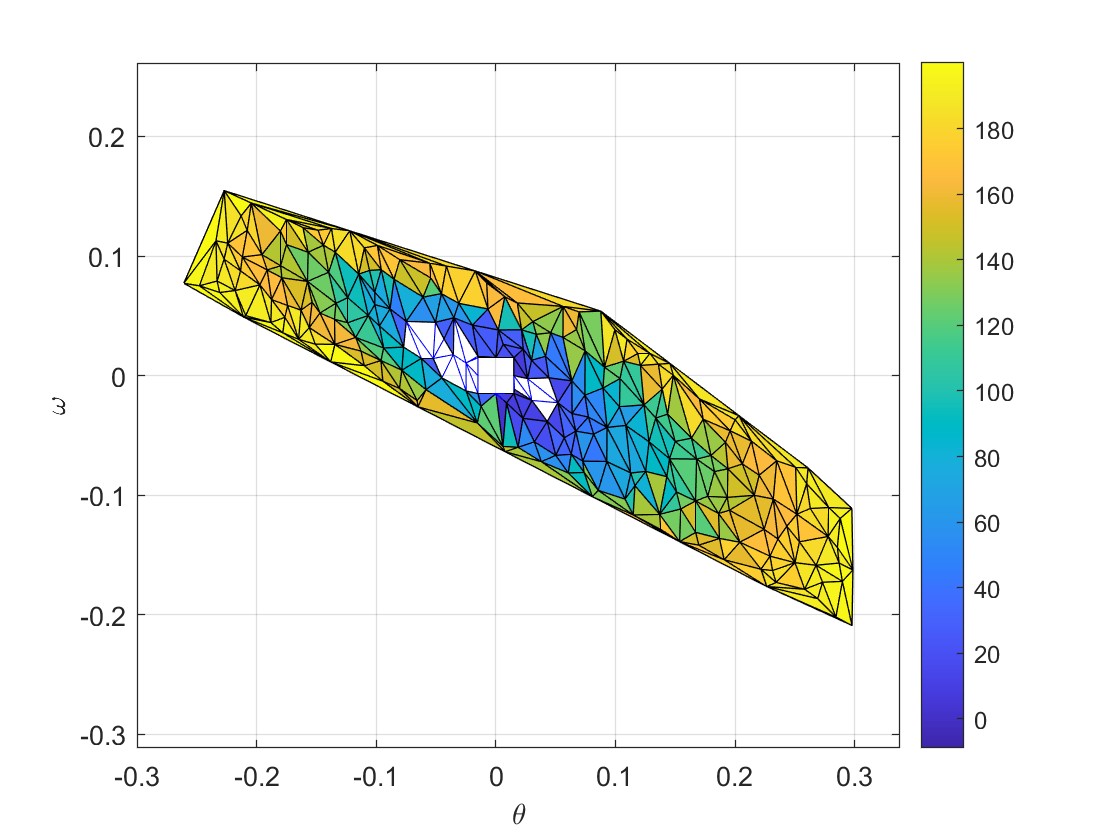}
        \caption{\tiny $N_{d_k}=370$, $N_{v_k}=239$}
    \end{subfigure}
    \begin{subfigure}{.42\linewidth}
        \includegraphics[width=\linewidth]{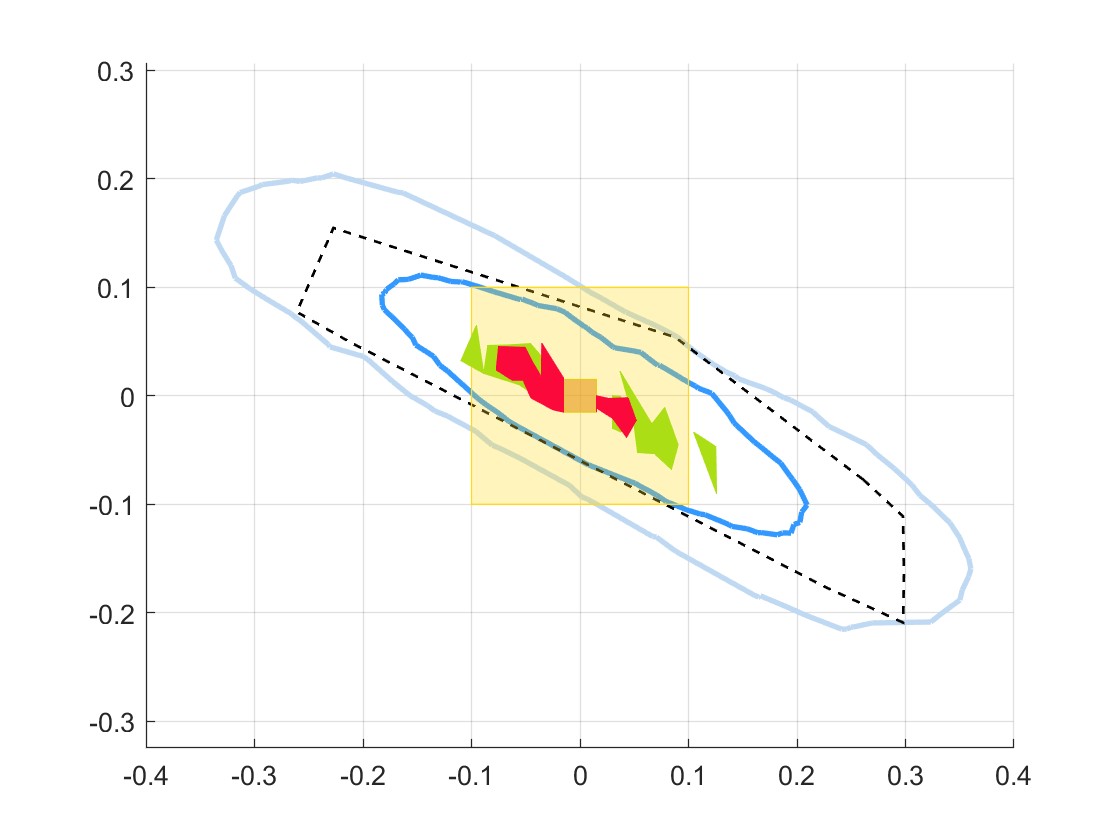}
        \caption{\tiny $\alpha_k = 145$, $\beta_k=2.32$}
    \end{subfigure}
    \caption{$k=4$}
    \label{fig:iter5}
\end{figure}

Each figure represents the results of an iteration $k$. The left sub-figure represents the Lyapunov function $V_k$ solution of the optimisation problem. The right sub-figure represents the following:
\begin{itemize}
    \item in a dashed line the borders of $\cX_k$;
    \item in a blue line the patchy region of attraction (delimited by the border of the level set $\mathbf{L}^{V_k}_{\alpha_k}$);
    \item in a faded blue line the border of the previous level set $\mathbf{L}^{V_{k-1}}_{\alpha_{k-1}}$;
    \item in red, the area $\cC_k$ that still needs to be certified;
    \item in green the previous uncertified area $\cC_{k-1}$;
    \item in orange the set $\cA_k$;
    \item and in transparent yellow the set $\cA=\cA_0$.
\end{itemize}

One can notice that the new level set $\mathbf{L}^{V_k}_{\alpha_k}$ inside $\cX_k$ in each iteration $k$ includes the region $\cC_{k-1}$ not certified by the Lyapunov candidate in the previous step, which is in line with the first condition of Theorem~\ref{thm}: $\cC_{k-1} \subset \mathbf{L}^{V_k}_{\alpha_k} \subset \cX_k \subset \mathbf{L}^{V_{k-1}}_{\alpha_{k-1}}$. One can also see that $\cC_k$ covers less area than $\cC_{k-1}$, hinting at algorithm convergence. $\cC_4$ in particular is small enough to validate the stopping criterion, which is the inclusion in $\cA$, in line with the validation of the second condition of Theorem~\ref{thm}: $\cC_K \subset \cA=\cA_0$. The validation of the two conditions implies that the level set $\mathbf{L}^{V_0}_{\alpha_0}$ is a region of attraction.

Figure~\ref{fig:levelsets} shows the progression of the level sets from figures~\ref{fig:iter1} to~\ref{fig:iter5} over the iterations (from green to blue). All trajectories starting in an invariant set converge to the subset inside it. As for the last invariant set, the trajectories converge to $\cA$ and then to the equilibrium point. In grey dotted lines, we can see the level sets obtained from a completely data-driven simulation. It converged to $\cA$ successfully, albeit in one more iteration compared to the informed simulation above. This proves that it is possible to certify the attraction of the equilibrium point in the region without relying on the model, and that the choice of tessellation strongly influences the speed of convergence.

\begin{figure}[htbp]
    \centering
    \includegraphics[width=.9\linewidth]{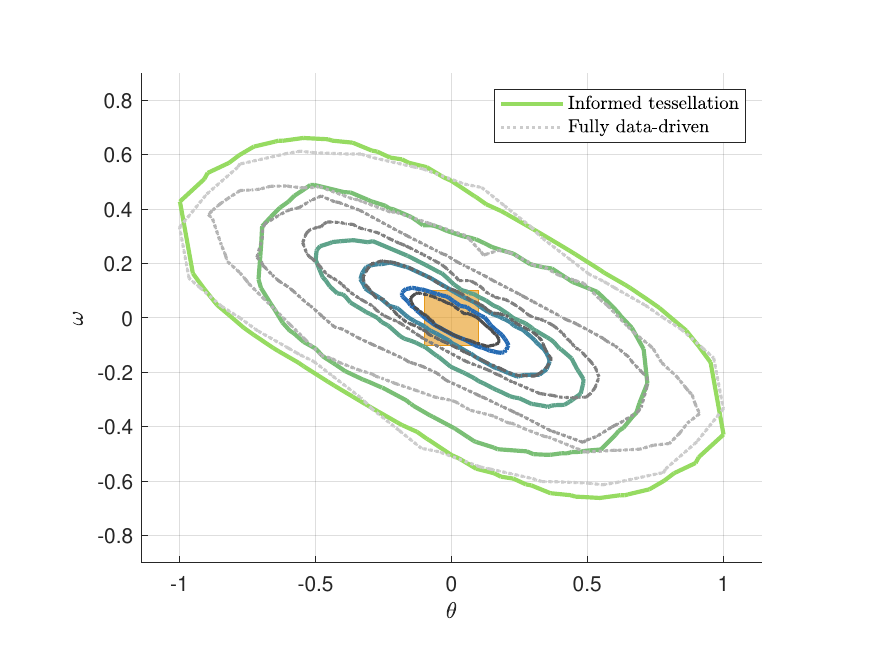}
    \caption{Graphical description of the relationship between the level sets' boundaries obtained iteratively. The progression in terms of set inclusions illustrates the convergence of the proposed algorithm.}
    \label{fig:levelsets}
\end{figure}

\section{Conclusion}
The paper concentrates on the data-driven stability certification. It employs an optimisation-based technique to build a PWA Lyapunov function and, as a contribution, proves that iterative treatment of the regions to be analysed can provide an increased complexity of the tessellations in the areas where the data are informative on the dynamical evolution. Practically, the resulting certificates of convergence rely on patchy Lyapunov functions and the construction of regions of attraction is ordered in the sense of inclusion. Several directions remain open for future research. Considering data's influence on the results, it would be interesting to consider some system trajectories as input instead of just independent data points. An essential next step is considering discrete and controlled systems to broaden the approach's applicability.

\bibliographystyle{ieeetr}

\end{document}